\documentclass[12pt, reqno]{amsart}
\usepackage{graphics,amsmath,amssymb,epsfig,color}
\usepackage{amsfonts}
\usepackage{float} 
\usepackage{amsmath}
\usepackage[utf8]{inputenc}
\usepackage{booktabs} 
\usepackage{dcolumn}
\usepackage{subfigure}
\usepackage{numprint}
\usepackage{mathtools}
\usepackage{amssymb}
\usepackage{color}

\newtheorem{theorem}{Theorem}[section]

\newtheorem{Proposition}[theorem]{Proposition}
\newtheorem{Corollary}[theorem]{Corollary}

\newcommand{\bX}{\boldsymbol{X}}
\newcommand{\bx}{\boldsymbol{x}}
\newcommand{\bgamma}{\boldsymbol{\gamma}}

\newcommand{\ba}{\boldsymbol{a}}
\newcommand{\bb}{\boldsymbol{b}}
\newcommand{\bA}{\boldsymbol{A}}
\newcommand{\bB}{\boldsymbol{B}}

\newcommand{\bzero}{\boldsymbol{0}}

\newcommand{\transpose}{^{\text{\sffamily T}}}
\newcommand{\inverse}{^{-1}}

\newcommand{\Esp}[1]{\mathbb{E}\left[ #1 \right]}    
\newcommand{\Var}[1]{\text{Var}\left[ #1 \right]}    
\newcommand{\Cov}[1]{\text{Cov}\left[ #1 \right]}   
\DeclarePairedDelimiter\floor{\lfloor}{\rfloor}
\begin{document}

\title{Macro vs. Micro Methods \\ in Non-Life Claims Reserving \\ (an Econometric Perspective)}
\author{Arthur Charpentier \& Mathieu Pigeon	}
\date{February, 28th, 2016}
\thanks{Arthur Charpentier, UQAM-Quantact \& Universit\'e de Rennes 1, France.\\
~~~~ Mathieu Pigeon, UQAM-Quantact.}

\maketitle
%%%%%%%%%%%%%%%%%%%%%%%%%%%%%%%%%%%%%%%%%%

\section{Introduction}

\subsection{Macro and Micro Methods}

For more than a century, actuaries have been using run-off triangles to project future payments, in non-life insurance. In the 30's, \cite{Astesan} formalized this technique that originated the popular {\em chain ladder} estimate. In the 90's, \cite{Mack} proved that the {\em chain ladder} estimate can be motivated by a simple stochastic model, and later on \cite{EnglandVerrall} provided a comprehensive overview on stochastic models that can be connected with the {\em chain ladder} method, included regression models, that could be seen as extension of the so-called "factor" methods used in the 70's.

But using the terminology of \cite{vanEeghen} and \cite{Taylor}, those were macro-level models for reserving. In the 70's, \cite{Karlsson} suggested to used some marked point process of
claims to project future payments, and quantify the reserves. More recently, \cite{Arjas},\cite{Jewell}, \cite{Norberg93}, \cite{Hesselager} or \cite{Norberg99} (among many others) investigated further some probabilistic micro-level models. These models handle claims related data on an individual basis,
rather than aggregating by underwriting year and development period. As mentioned in \cite{HESSELAGERVERRALL}, these methods have not (yet) found great popularity in practice, since they are more difficult to apply.

All macro-level models are based on aggregate data found in a run-off triangle, which is their strength, but also probably their weakness. They are easy to understand, and can be mentioned in financial communication, without disclosing too much information. From a computational perspective, those models can also be implemented in a single spreadsheet. But recently, \cite{EnglandVerrall} started questioning actuaries about possible use of detailed micro-level information. Those models can incorporate heterogeneity, structural changes, etc (see \cite{Friedland} for a discussion).

\subsection{Best estimates and Variability}

In the context of macro-level models, \cite{EnglandVerrall} mention that prediction errors can be large, because of the small number of {\em observations} used in run-off triangles. Quantifying uncertainty in claim reserving methods is not only important in actuarial practice and to assess accuracy of predictive models, it is also a regulatory issue.

\cite{Ant} and \cite{Pig2} obtained, on real data analysis, lower variance on the total amount of reserves with "micro" models than with "macro" ones. A natural question is about the generality of such result. Should "micro" model generate less variability than standard "macro" ones? That is the question that initiated that paper.

\subsection{Agenda}

In section~\ref{sec2}, we detail intuitive results we expect when
aggregating data by clusters, moving from micro-level models to
macro-level ones. More precisely, we explain why with a linear model
and a Poisson regression, macro- and micro-level models are
equivalent. We also discuss the case of the Poisson regression model
with random intercept. In section~\ref{sec3}, we study "micro" and "macro" models in the context of claims reserving, on real data, as well as simulated ones.

%%%%%%%%%%%%%%%%%%%%%%%%%%%%%%%%%%%%%%%%%%

\section{Clustering in Generalized Linear Mixed Models}\label{sec2}

In the economic literature, several papers discuss the use of "micro" vs. "macro" data, for instance in the context of unemployment duration in \cite{vandenBerga} or in the context of inflation in \cite{Altissimo}. In \cite{vandenBerga}, it is mentioned that both models are interesting, since "micro" data can be used to capture heterogeneity while "macro" data can capture cycle and more structural patterns. In \cite{Altissimo}, it is demonstrated that both heterogeneity and aggregation might explain the persistence of inflation at the macroeconomic level. 

In order to clarify notation, and make sure that objects are well defined, we use small letters for sample values, e.g. $y_i$, and capital letters for underlying random variables, e.g. $Y_i$ in the sense that $y_i$ is a realisation of random variable $Y_i$. Hence, in the case of the linear model (see Section \ref{subsec:linear}), we usually assume that $Y_i\sim\mathcal{N}(\boldsymbol{x}_i\transpose\bb,\sigma^2)$, and then
$\widehat{\boldsymbol{b}}$ is the estimated model, in the sense that $\widehat{\boldsymbol{b}}=(\boldsymbol{x}\boldsymbol{x}\transpose)\inverse\boldsymbol{x}\boldsymbol{y}$ while $\widehat{\boldsymbol{B}}=(\boldsymbol{x}\boldsymbol{x}\transpose)\inverse\boldsymbol{x} \boldsymbol{Y}$ (here covariates $\boldsymbol{x}$ are given, and non stochastic). Since $\widehat{\boldsymbol{B}}$ is seen as a random variable, we can write $\mathbb{E}[\widehat{\boldsymbol{B}}]=\boldsymbol{b}$.

With a Poisson regression, $Y_i\sim\mathcal{P}(\lambda_i)$ with $\lambda_i=\exp[\boldsymbol{x}_i\transpose\boldsymbol{b}]$. In that case,
$\text{Var}[Y_i]=\mathbb{E}[Y_i]=\lambda_i$. The estimate parameter $\widehat{\boldsymbol{b}}$ is a function of the observations, $(\boldsymbol{x}_i,\boldsymbol{y}_i)$'s, while $\widehat{\boldsymbol{B}}$ is a function of the observations, $(\boldsymbol{x}_i,\boldsymbol{Y}_i)$'s. In the context of the Poisson regression, recall that $\mathbb{E}[\widehat{\boldsymbol{B}}]\rightarrow \boldsymbol{b}$ as $n$ goes to infinity. With a quasi-Poisson regression, $Y_i$ does not have, per se, a proper distribution. Nevertheless, its moments are well defined, in the sense that $\text{Var}[Y_i]=\varphi\mathbb{E}[Y_i]=\varphi\lambda_i$. And for convenience, we will denote  $Y_i\sim q\mathcal{P}(\lambda_i)$, with an abuse of notation.

In this section, we will derive some theoretical results regarding aggregation in econometric models. 

\subsection{The multiple linear regression model}\label{subsec:linear}

Consider a (multiple) linear regression model,

\begin{align}\label{eq1}
y_{i,g} &= \bx_{g}\transpose\ba + \varepsilon_{i,g},\\
\ba &= \begin{bmatrix} a_1 & \ldots & a_{k+1} \end{bmatrix}\transpose \qquad \bx_{g} = \begin{bmatrix} x_{g,1} & \ldots & x_{g, k+1}\end{bmatrix}\transpose,\nonumber
\end{align}

where observations belong to a cluster $g$ and are indexed by $i$ within a cluster $g$, $i = 1, \ldots, n_g$, $g = 1, \ldots, m$. 
Assume further assumptions of the classical linear regression model \cite{Greene}, i.e.,
\begin{itemize}
\item[(LRM1)] no multicollinearity in the data matrix;
\item[(LRM2)] exogeneity of the independent variables $\Esp{\varepsilon_{i,g} | \bx_g} = 0$, $i = 1, \ldots, n_g$, $g = 1, \ldots , m$; and
\item[(LRM3)] homoscedasticity and nonautocorrelation of error terms with $\Var{\varepsilon_{i,g}} = \sigma^2$.
\end{itemize}
Stacking observations within a cluster yield the following model

\begin{align}\label{eq2}
    \overline{y}_{g} &= \bx_{g}\transpose\bb+e_{g},\intertext{where}
    \overline{y}_g &= \frac{1}{n_g}\sum_i y_{i,g} \text{ and }
\bb = \begin{bmatrix} b_1 & \ldots & b_{k+1} \end{bmatrix}\transpose \nonumber
\end{align}

with similar assumptions except for $\Var{e_g} = \sigma^2/n_g$. Those two models are equivalent, in the sense that the following proposition holds.

\begin{Proposition}
Model~(\ref{eq1}) on a micro level and model~(\ref{eq2}) on a macro level are equivalent, in the sense that
\begin{itemize}%[leftmargin=*,labelsep=4mm]
\item[(i)] $\widehat{\ba}_{OLS}=\widehat{\bb}_{OLS}$ when weights $n_g$ are used in model~(\ref{eq2}); and
\item[(ii)] $\displaystyle{\sum_{i,g} \widehat{y}_{i,g} = \sum_{g} \widehat{y}_{g}}$ where $y_g =n_g \overline{y}_g$.
\end{itemize}
\end{Proposition}
\begin{proof}
\begin{itemize}%[leftmargin=*,labelsep=4mm]
\item[(i)]
The ordinary least-squares estimator for $\ba$ - from
model~(\ref{eq1}) - is defined as

\begin{equation}
\widehat{\ba} = \underset{\ba}{\text{argmin}}\left\{
\sum_{i,g} \left(y_{i,g}-
\bx_{g}\transpose\ba\right)^2
\right\}
\end{equation}

which can also be written

\begin{equation}
\widehat{\ba} = \underset{\ba}{\text{argmin}}\left\{
\sum_{i,g} \left(y_{i,g}-\overline{y}_g+\overline{y}_g
-\bx_{g}\transpose\ba\right)^2
\right\}.
\end{equation}

Now, observe that

\begin{align*}
\sum_{i,g}  \left(y_{i,g}-\overline{y}_g+\overline{y}_g-
\bx_{g}\transpose\ba\right)^2
&=
\sum_{i,g}  (y_{i,g}-\overline{y}_g)^2+(\overline{y}_g-
\bX_{g}\transpose\ba)^2\\
&\phantom{=}+ 2 (y_{i,g}-\overline{y}_g)(\overline{y}_g-
\bx_{g}\transpose\ba),
\end{align*}

where the first term is independent of $\ba$ (and can be removed from the optimization program), and the term with cross-elements sums to 0. Hence,

\begin{equation}
\widehat{\ba} = \underset{\ba}{\text{argmin}}\left\{
\sum_{i,g} (\overline{y}_g-
\bx_{g}\transpose\ba)^2
\right\}=
\underset{\ba}{\text{argmin}}\left\{
\sum_{g} n_g (\overline{y}_g-
\bx_{g}\transpose\ba)^2
\right\}=\widehat{\bb},
\end{equation}

where $\widehat{\bb}$ is the least square estimator of $\bb$ from model~(\ref{eq2}), when weights $n_g$ are considered.
\item[(ii)] If we consider the sum of predicted values, observe that

\begin{equation}
\sum_{i,g} \widehat{y}_{i,g} = \sum_{g} n_g \bx_{g}\transpose\widehat{\ba} = 
\sum_{g} n_g \underbrace{\bx_{g}\transpose\widehat{\bb}}_{\widehat{\overline{y}}_{g}}=\sum_{g} \widehat{y}_{g} .
\end{equation}

Hence, the sum of predictions obtained from model~(\ref{eq1}) is the same as the sum of predictions obtained from model~(\ref{eq2}), even if partial sums are considered. 
\end{itemize}
\end{proof}
In the proposition above, the equality should be understood as the equality between estimators. Hence we have the following corollary.
\begin{Corollary}
We define the following matrices

\begin{align*}
\boldsymbol{Y}_g &= \begin{bmatrix} Y_{1,g} & \ldots & Y_{n_g, g} \end{bmatrix}\transpose \qquad \boldsymbol{Y} = \begin{bmatrix} \boldsymbol{Y}_1\transpose & \ldots & \boldsymbol{Y}_m\transpose \end{bmatrix}\transpose \qquad \overline{\boldsymbol{Y}} = \begin{bmatrix} \overline{Y}_1 & \ldots & \overline{Y}_m \end{bmatrix}\transpose\\
\boldsymbol{\varepsilon}_g &= \begin{bmatrix} \varepsilon_{1,g} & \ldots & \varepsilon_{n_g, g} \end{bmatrix}\transpose \qquad \boldsymbol{\varepsilon} = \begin{bmatrix} \boldsymbol{\varepsilon}_1 & \ldots & \boldsymbol{\varepsilon}_m \end{bmatrix}\transpose \qquad \boldsymbol{e} = \begin{bmatrix} e_1 & \ldots & e_m \end{bmatrix}\transpose\\
\bx_g^{n_g} &= \underbrace{\begin{bmatrix} \bx_g & \ldots & \bx_g \end{bmatrix}}_{\text{$n_g$ times}} \qquad \bx = \begin{bmatrix} \bx_1^{n_1} & \ldots & \bx_{m}^{n_m} \end{bmatrix} \qquad \overline{\bx} = \begin{bmatrix} \bx_1^1 & \ldots & \bx_m^1 \end{bmatrix},
\end{align*}

the $(1 \times n_g)$ vectors $\boldsymbol{1}_{n_g} = \begin{bmatrix} 1 & \ldots & 1 \end{bmatrix}$ and $\boldsymbol{0}_{n_g} = \begin{bmatrix} 0 & \ldots & 0 \end{bmatrix}$, and the matrix

\begin{align*}
    \boldsymbol{1} &= \begin{bmatrix}
       \boldsymbol{1}_{n_1} & \boldsymbol{0}_{n_2} & \ldots & \boldsymbol{0}_{n_m}\\
        \boldsymbol{0}_{n_1} & \boldsymbol{1}_{n_2} & \ldots & \boldsymbol{0}_{n_m}\\
            \vdots & \vdots & \ddots & \vdots\\
                \boldsymbol{0}_{n_1} & \boldsymbol{0}_{n_2} &\ldots & \boldsymbol{1}_{n_m}
    \end{bmatrix}.
\end{align*}

The OLS estimators are given by

\begin{align*}
\widehat{\bA}_{OLS} &= \underset{\ba}{\text{argmin}}\lbrace \left(\boldsymbol{Y} - \bx\transpose\ba\right)\transpose\left(\boldsymbol{Y} - \bx\transpose\ba\right)\rbrace\\
&= \left(\bx\bx\transpose\right)\inverse\bx \boldsymbol{Y}\\
\widehat{\bB}_{OLS} &= \underset{\bb}{\text{argmin}}\lbrace \left((\boldsymbol{\sqrt{\boldsymbol{1}\boldsymbol{1}\transpose}})\overline{\boldsymbol{Y}} - \boldsymbol{\sqrt{\boldsymbol{1}\boldsymbol{1}\transpose}}\overline{\bx}\transpose\bb\right)\transpose\left((\boldsymbol{\sqrt{\boldsymbol{1}\boldsymbol{1}\transpose}})\overline{\boldsymbol{Y}} - (\boldsymbol{\sqrt{\boldsymbol{1}\boldsymbol{1}\transpose}})\overline{\bx}\transpose\bb\right)\rbrace \\
&= \left(\overline{\bx}\boldsymbol{1}\boldsymbol{1}\transpose\overline{\bx}\transpose\right)\inverse\overline{\bx}\boldsymbol{1}\boldsymbol{1}\transpose\overline{\boldsymbol{Y}}.
\end{align*}

Model~(\ref{eq1}) on a micro level and model~(\ref{eq2}) on a macro level are equivalent, in the sense that
\begin{itemize}%[leftmargin=*,labelsep=4mm]
\item[(i)] $\mathbb{E}[\widehat{\bA}_{OLS}]=\mathbb{E}[\widehat{\bB}_{OLS}]$ and $\text{\em Var}[\widehat{\bA}_{OLS}]=\text{\em Var}[\widehat{\bB}_{OLS}]$, when weights $n_g$ are used in model~(\ref{eq2}); and
\item[(ii)] $\mathbb{E}\displaystyle{\left[\sum_{i,g} \widehat{Y}_{i,g}\right] = \mathbb{E}\left[\sum_{g} \widehat{Y}_{g}\right]}$ and $\text{\em Var}\displaystyle{\left[\sum_{i,g} \widehat{Y}_{i,g}\right] = \text{\em Var}\left[\sum_{g} \widehat{Y}_{g}\right]}$.
\end{itemize}
\end{Corollary}
\begin{proof}
Straightforward calculations lead to $(\boldsymbol{1}\boldsymbol{1}\transpose)\inverse\boldsymbol{1}\boldsymbol{Y} = \overline{\boldsymbol{Y}}$ and $\overline{\bx}\boldsymbol{1} = \bx$.
\begin{itemize}%[leftmargin=*,labelsep=4mm]
\item[(i)] Let

\begin{align*}
\Esp{\widehat{\bB}_{OLS}} &= \left(\overline{\bx} \boldsymbol{\boldsymbol{1}\boldsymbol{1}\transpose} \overline{\bx}\transpose\right)\inverse \overline{\bx}\boldsymbol{\boldsymbol{1}\boldsymbol{1}\transpose}\Esp{\overline{\boldsymbol{Y}}}\\
&= \left(\overline{\bx} \boldsymbol{\boldsymbol{1}\boldsymbol{1}\transpose} \overline{\bx}\transpose\right)\inverse \overline{\bx}\boldsymbol{\boldsymbol{1}\boldsymbol{1}\transpose}(\boldsymbol{1}\boldsymbol{1}\transpose)\inverse\boldsymbol{1}\Esp{\boldsymbol{Y}}\\
&= \left(\bx \bx\transpose\right)\inverse \bx\Esp{\boldsymbol{Y}} = \Esp{\widehat{\bA}_{OLS}}.
\end{align*}

For the equality of variances, we have

\begin{align*}
    &\Var{\widehat{\bB}_{OLS}}\\ &=\left(\overline{\bx}\boldsymbol{1}\boldsymbol{1}\transpose\overline{\bx}\transpose\right)\inverse\overline{\bx}\boldsymbol{1}\boldsymbol{1}\transpose\Var{\overline{\boldsymbol{Y}}}\left(\left(\overline{\bx}\boldsymbol{1}\boldsymbol{1}\transpose\overline{\bx}\transpose\right)\inverse\overline{\bx}\boldsymbol{1}\boldsymbol{1}\transpose\right)\transpose\\
    &= \left(\overline{\bx}\boldsymbol{1}\boldsymbol{1}\transpose\overline{\bx}\transpose\right)\inverse\overline{\bx}\boldsymbol{1}\boldsymbol{1}\transpose(\boldsymbol{1}\boldsymbol{1}\transpose)\inverse\boldsymbol{1}\Var{\boldsymbol{Y}}\boldsymbol{1}\transpose \left((\boldsymbol{1}\boldsymbol{1}\transpose)\inverse\right)\transpose\left(\left(\overline{\bx}\boldsymbol{1}\boldsymbol{1}\transpose\overline{\bx}\transpose\right)\inverse\overline{\bx}\boldsymbol{1}\boldsymbol{1}\transpose\right)\transpose\\
    &= (\bx\bx\transpose)\inverse\bx \Var{\boldsymbol{Y}}\bx\transpose\left((\bx\bx\transpose)\inverse\right)\transpose\\
    &= \Var{\widehat{\bA}_{OLS}}.
\end{align*}

\item[(ii)] Let

\begin{align*}
\Esp{\sum_{g} \widehat{Y}_{g}} &= \Esp{\boldsymbol{1}_m\boldsymbol{1}\boldsymbol{1}\transpose \widehat{\overline{\boldsymbol{Y}}}} = \Esp{\boldsymbol{1}_m\boldsymbol{1}\boldsymbol{1}\transpose \overline{\bx}\transpose\widehat{\bB}}\\
&= \Esp{\boldsymbol{1}_m\boldsymbol{1}\boldsymbol{1}\transpose \overline{\bx}\transpose\widehat{\bA}} = \Esp{\boldsymbol{1}_n \bx\transpose \widehat{\bA}}\\
&= \Esp{\boldsymbol{1}_n\widehat{\boldsymbol{Y}}} = \Esp{\sum_{i,g} Y_{i,g}}.
\end{align*}

The proof of the equality of variances is similar.
\end{itemize}
\end{proof}

\subsection{The quasi-Poisson regression}\label{ss:quasi}

A similar result can be obtained in the context of Poisson regressions. 
A generalized linear model \cite{McCullaghNelder} is made up of a linear predictor $\bx\transpose\bb$, a link function that describes how the expected value depends on this linear predictor and a variance function that describes how the variance, depends on the expected value $\Var{Y} = \varphi V\left(\Esp{Y}\right)$, where $\varphi$ denotes the dispersion parameter. For the Poisson model, the variance is equal to the mean, i.e., $\varphi = 1$ and $V(\Esp{Y}) = \Esp{Y}$. This may be too restrictive for many actuarial illustrations, which often show more variation than given by expected values. We use the term \emph{over-dispersed} for a model where the variance exceeds the expected value. A common way to deal with over-dispersion is a \emph{quasi-likelihood} approach (see \cite{McCullaghNelder} for further discussion) where a model is characterized by its first two moments.

Consider either a Poisson regression model, or a quasi-Poisson one,

\begin{equation}\label{eq3}
Y_{i,g} \sim \mathcal{P}(\lambda_{i,g})\text{ or }Y_{i,g} \sim q\mathcal{P}(\lambda_{i,g}).
\end{equation}

In the case of a Poisson regression, 

\begin{equation}
\Esp{Y_{i,g}} = \lambda_{i,g}=\exp[\bx_{g}\transpose\ba + \ln(1/n_g)]
\text{ and }\Var{Y_{i,g}} = \lambda_{i,g}, 
\end{equation}

and in the context of a quasi-Poisson regression,

\begin{equation}
\Esp{Y_{i,g}}= \lambda_{i,g}=\exp[\bx_{g}\transpose\ba + \ln(1/n_g)]
\text{ and }\Var{Y_{i,g}} = \varphi_{\text{micro}}\lambda_{i,g}, 
\end{equation}

with $\varphi_{\text{micro}}>0$ for a quasi-Poisson regression ($\varphi_{\text{micro}}>1$ for overdispersion). 
Here again, stacking observations within a cluster yield the following
model (on the sum and not the average value, to have a valid
interpretation with a Poisson distribution)

\begin{equation}\label{eq4}
Y_{g}= \sum_i Y_{i,g} \sim \mathcal{P}(\lambda_{g})\text{ or }Y_{g} \sim q\mathcal{P}(\lambda_{g}).
\end{equation}

In the context of a Poisson regression,  

\begin{equation}
\Esp{Y_g} = \lambda_{g} =\exp[\bx_{g}\transpose\bb]\text{ and }\Var{Y_g}= \lambda_g\nonumber,
\end{equation}

and in the context of a quasi-Poisson regression,

\begin{equation}
\Esp{Y_{g}} = \lambda_{g}=\exp[\bx_{g}\transpose\bb] \text{ and
}\Var{Y_{g}} = \varphi_{\text{macro}}\lambda_{g}, 
\end{equation}

with $\varphi_{\text{macro}}>0$ for a quasi-Poisson regression. 
Here again, those two models ("micro" and "macro") are equivalent, in the sense that the following proposition holds.

\begin{Proposition}\label{propPoisson}
Model~(\ref{eq3}) on a micro level and model~(\ref{eq4}) on a macro level are equivalent in the sense that
\begin{itemize}%[leftmargin=*,labelsep=4mm]
\item[(i)] $\widehat{\ba}_{MLE}=\widehat{\bb}_{MLE}$; and

\item[(ii)] $\displaystyle{\sum_{i,g} \widehat{y}_{i,g} = \sum_{g} \widehat{y}_{g}}$.
\end{itemize}
\end{Proposition}
\begin{proof}
  \begin{itemize}%[leftmargin=*,labelsep=4mm]
  \item[(i)]
    Maximum likelihood estimator of $\ba$ is the solution of

    \begin{align*}
      \sum_{i,g} \left(\frac{y_{i,g} - \exp[\bx_{g}\transpose\ba]}{\varphi_{\text{micro}}}\right)\bx_g &= \bzero
    \end{align*}

    or equivalently

    \begin{align*}
      \sum_{i,g} \left(y_{i,g} - \exp[\bx_{g}\transpose\ba]\right)\bx_g &= \bzero
    \end{align*}

    With offsets $\lambda_{g}^*=\exp[\bx_{g}\transpose\bb +
    \log(n_g)]$, $g = 1, \ldots, m$, maximum likelihood estimator of
    $\bb$ is the solution (as previously, we can remove
    $\varphi_{\text{macro}}$) of

    \begin{align*}
      \sum_{g}\left(y_g - n_g\exp[\bx_{g}\transpose\bb]\right)\bx_g &= 0\\
      \sum_{i,g}\left(y_{i,g} - \exp[\bx_{g}\transpose\bb]\right)\bx_g &= 0.
    \end{align*}

    Hence, $\widehat{\ba} = \widehat{\bb}$, as (unique) solutions of
    the same system of equations.
  \item[(ii)] 
    The sum of predicted values is

\begin{align*}
      \sum_{i,g} \widehat{y}_{i,g} &= \sum_g
      n_g\widehat{\lambda}_{i,g} 
= \sum_g n_g\exp[\bx_{g}\transpose\widehat{\ba}] = \sum_g n_g\exp[\bx_{g}\transpose\widehat{\bb}]\\
      &= \sum_g \exp[\bx_{g}\transpose\widehat{\bb} + \log(n_g)] = \sum_g \widehat{\lambda}_g^* = \sum_g \widehat{y}_g.
    \end{align*}

\textcolor{white}{---}
  \end{itemize}
\end{proof}
Nevertheless, as we will see later on, the Corollary obtained in the context of a Gaussian linear model does not hold in the context of a quasi-Poisson regression.
\begin{Corollary}\label{cor:1}
Model~(\ref{eq3}) on a micro level and model~(\ref{eq4}) on a macro level are asymptotically equivalent for Poisson regressions, in the sense that
\begin{itemize}%[leftmargin=*,labelsep=4mm]
\item[(i)] $\mathbb{E}[\widehat{\bA}_{MLE}]=\mathbb{E}[\widehat{\bB}_{MLE}]$ and $\text{\em Var}[\widehat{\bA}_{MLE}]=\text{\em Var}[\widehat{\bB}_{MLE}]$, when $n$ goes to infinity; and
\item[(ii)] $\mathbb{E}\displaystyle{\left[\sum_{i,g} \widehat{Y}_{i,g}\right] = \mathbb{E}\left[\sum_{g} \widehat{Y}_{g}\right]}$ and $\text{\em Var}\displaystyle{\left[\sum_{i,g} \widehat{Y}_{i,g}\right] = \text{\em Var}\left[\sum_{g} \widehat{Y}_{g}\right]}$, when $n$ goes to infinity.
\end{itemize}
\end{Corollary}
\begin{proof}
  \begin{itemize}%[leftmargin=*,labelsep=4mm]
  \item[(i)] A classical result of asymptotic theory for maximum likelihood estimators indicates that, under mild regularity conditions, $\Esp{\widehat{\bA}_{MLE}} \to \ba$ and $\Esp{\widehat{\bB}_{MLE}} \to \bb$ as $n \to \infty$. Without "micro" covariates in the model, $\ba = \bb$ and so, $\Esp{\widehat{\bB}_{MLE}} = \Esp{\widehat{\bA}_{MLE}}$ when $n \to \infty$. For model (\ref{eq3}), the Fisher information matrix is $\boldsymbol{I}(\bA) = \bx \boldsymbol{W} \bx\transpose$ and, when $n \to \infty$, $\Var{\widehat{\bA}} \to \left(\bx \boldsymbol{W} \bx\transpose\right)\inverse$, where $\boldsymbol{W} = \text{diag}((\lambda_1/n_1) \boldsymbol{1}_{n_1},\ldots,(\lambda_m/n_m)\boldsymbol{1}_{n_m})$. For model (\ref{eq4}), we have $\boldsymbol{I}(\bB) = \overline{\bx} \boldsymbol{1}\boldsymbol{W}\boldsymbol{1}\transpose\overline{\bx}\transpose =
    \bx \boldsymbol{W} \bx\transpose$ and, when $n \to \infty$, $\Var{\widehat{\bB}} \to \left(\bx \boldsymbol{W} \bx\transpose\right)\inverse$.
  \item[(ii)] By using a similar argument, we have when $n$ goes to
    infinity

    \begin{align*}
      \Esp{\sum_g \widehat{Y}_g} &= \Esp{\boldsymbol{1}_m\boldsymbol{1}\boldsymbol{1}\transpose\widehat{\overline{\boldsymbol{Y}}}} = \boldsymbol{1}_m\boldsymbol{1}\boldsymbol{1}\transpose M_{\widehat{\bB}}\left(\overline{\bx}\transpose\right)\\
     &= \boldsymbol{1}_m\boldsymbol{1}\boldsymbol{1}\transpose M_{\widehat{\bA}}\left(\overline{\bx}\transpose\right) = \Esp{\boldsymbol{1}_n\boldsymbol{1}\transpose e^{\overline{\bx}\transpose\widehat{\bA}}} = \Esp{\boldsymbol{1}_ne^{\bx\transpose\widehat{\bA}}}\\
      &= \Esp{\boldsymbol{1}_n\widehat{\boldsymbol{Y}}} = \Esp{\sum_{i,g} \widehat{Y}_{g,i}}.
    \end{align*}

\textcolor{white}{---}
  \end{itemize}
\end{proof}
In small or moderate-sized samples, it should be noted that $\widehat{\bA}$ and $\widehat{\bB}$ may be biased for $\bA$ and $\bB$, respectively. Generally, this bias is negligible compared with the standard errors (see \cite{CMC} and \cite{MW}).

In the quasi-Poisson micro-level model (from model~\eqref{eq3}), as
discussed above, the estimator of $\ba$ is the solution of the
quasi-score function

\begin{align*}  
    \sum_{i,g} \left(\frac{y_{i,g} - \lambda_{i,g}}{\varphi_{micro}}\right)\bx_{g} &= 0,
\end{align*}

which implies $\widehat{\ba}_{QLE}=\widehat{\ba}_{MLE}$. The classical
Pearson estimator for the dispersion parameter
$\varphi_{\text{micro}}$ is
 
\begin{align*}
    \widehat{\varphi}_{micro} &= \sum_{i,g}\frac{\left(y_{i,g} - \widehat{y}_{i,g}\right)^2/\widehat{y}_{i,g}}{\sum_gn_g - (k+1)} 
    = \sum_{i,g} \frac{\left(n_gy_{i,g} - \widehat{y}_{g}\right)^2 / n_g\widehat{y}_{g}}{\sum_g n_g - (k+1)}.
\end{align*}

Empirical evidence (see \cite{RUG}) support the use of the Pearson
estimator for estimating $\varphi$ because it is the most robust
against the distributional assumption. In a similar way, the
quasi-Poisson macro-level model (from model~\eqref{eq4}), the
estimator of $\bb$ is the solution of 

\begin{align*}  
    \sum_{g} \left(\frac{y_{g} - \lambda_{g}}{\varphi_{macro}}\right)\bx_{g} &= 0,
\end{align*}

which implies here also $\widehat{\bb}_{QLE}=\widehat{\bb}_{MLE}$. The dispersion parameter $\varphi$ is estimated by

\begin{align*}
 \widehat{\varphi}_{macro} &= \sum_g \frac{\left(y_g - \widehat{y}_g\right)^2/\widehat{y}_g}{m - (k+1)}.
\end{align*}

Clearly, $\widehat{\varphi}_{micro} \ne \widehat{\varphi}_{macro}$ involving the following results.
\begin{Corollary}\label{cor:2}
Model~(\ref{eq3}) on a micro level and model~(\ref{eq4}) on a macro level are not asymptotically equivalent for quasi-Poisson regressions, in the sense that
\begin{itemize}%[leftmargin=*,labelsep=4mm]
\item[(i)] $\mathbb{E}[\widehat{\bA}_{QLE}]=\mathbb{E}[\widehat{\bB}_{QLE}]$ but $\text{\em Var}[\widehat{\bA}_{QLE}] \neq \text{\em Var}[\widehat{\bB}_{QLE}]$, when $n$ goes to infinity; and
\item[(ii)] $\mathbb{E}\displaystyle{\left[\sum_{i,g} \widehat{Y}_{i,g}\right] = \mathbb{E}\left[\sum_{g} \widehat{Y}_{g}\right]}$ but $\text{\em Var}\displaystyle{\left[\sum_{i,g} \widehat{Y}_{i,g}\right] \neq \text{\em Var}\left[\sum_{g} \widehat{Y}_{g}\right]}$, when $n$ goes to infinity.
\end{itemize}
\end{Corollary}
\begin{proof}
\begin{itemize}%[leftmargin=*,labelsep=4mm]
\item[(i)] The property that variances are not equal is a direct consequence of classical results from the theory of generalized linear models (see \cite{McCullaghNelder}), since the covariance matrices of estimators are given by

\begin{align}\label{eq:W}
    \Var{\widehat{\bB}} &\to \widehat{\varphi}_{macro}\left(\bx \boldsymbol{W} \bx\transpose\right)^{-1} \nonumber\intertext{and}
        \Var{\widehat{\bA}} &\to \widehat{\varphi}_{micro}\left(\bx \boldsymbol{W} \bx\transpose\right)^{-1},
\end{align}

when $n$ goes to infinity. Thus, covariance matrices of estimators are asymptotically equal for the Poisson regression model but differ for the quasi-Poisson model because $\widehat{\varphi}_{micro} \ne \widehat{\varphi}_{macro}$. 
\item[(ii)] Since the MLE and the QLE share the same asymptotic distribution (see \cite{McCullaghNelder}), the proof is similar to \ref{cor:1}(ii).
\end{itemize}
\end{proof}

\subsection{Poisson regression with random effect}\label{ss:RE}

In the micro-level model described by Equation~\eqref{eq3},
observations made for the same event (subject) at different periods
are supposed to be independent. Within-subject correlation can be
included in the model by adding random, or subject-specific, effects
in the linear predictor. In the Poisson regression model with random
intercept, the between-subject variation is modeled by a random
intercept $\bgamma$ which represents the combined effects of all omitted covariates.

Let $Y_{g}^{(t)}$ represent the sum of all observations from subject $t$, in the cluster $g$ and

\begin{align*}
    \Esp{Y_{g}^{(t)}|\gamma_t} &= \exp[\bx_{g}\transpose\ba + \gamma_t]\\
    \Var{Y_{g}^{(t)}|\gamma_t} &= \Esp{Y_{g}^{(t)}|\gamma_t}\\
       \bgamma = \begin{bmatrix}
        \gamma_1 \\
        \ldots\\
        \gamma_T
       \end{bmatrix} &\sim \boldsymbol{N}_T\left(\boldsymbol{0}, \sigma^2\boldsymbol{I}\right),
\end{align*}

where $\boldsymbol{I}$ is the $(T \times T)$ identity matrix, and $\boldsymbol{N}_T(\boldsymbol{\mu},\boldsymbol{\Sigma})$ the $T$-dimensional Gaussian distribution with mean $\boldsymbol{\mu}$ and covariance matrix $\boldsymbol{\Sigma}$. Straightforward calculations lead to

\begin{align*}
\Esp{Y_{g}^{(t)}} &= \exp[\bx_{g}\transpose\ba + \sigma^2/2]\\
\Var{Y_{g}^{(t)}} &= \Esp{Y_{g}^{(t)}}\left(1 + \Esp{Y_{g}^{(t)}}\left(\exp[\sigma^2] - 1\right)\right).
\end{align*}

Hence,

\begin{equation}
\Var{Y_{g}^{(t)}} > \Esp{Y_{g}^{(t)}}, \quad \sigma^2 > 0.
\end{equation}

This last equation shows that the Poisson regression model with random intercept leads to an over-dispersed marginal distribution for the variable $Y_g^{(t)}$. The maximum likelihood estimation for parameters requires Laplace approximation and numerical integration (see the Chapter 4 of \cite{MuAn} for more details). This model is a special case of multilevel Poisson regression model and estimation can be performed with various statistical softwares such as HLM, SAS, Stata and R (with package \textbf{lme4}).

One may be interested to verify the need of a source of between-subject variation. Statistically, it is equivalent to testing the variance of $\gamma$ to be zero. In this particular case, the null hypothesis places $\sigma^2$ on the boundary of the model parameter space which complicates the evaluation of the asymptotic distribution of the classical likelihood ratio test (LRT) statistic. From the very general result of \cite{SL}, it can be demonstrated (see \cite{ZL}) that the asymptotic null distribution of the LRT statistic is a $50/50$ mixture of $\chi_0^2$ and $\chi_1^2$ as $\sum_g n_g \to \infty$. In this case, obtaining an equivalent macro-level model is of little practical interest since the construction of the variance-covariance matrix would require knowledge of the individual ("micro") data.

\section{Clustering and loss reserving models}\label{sec3}

A loss reserving macro-level model is constructed from data summarized in a table called run-off triangle. Aggregation is performed by occurrence and development periods (typically years). For occurrence period $i$, $i = 1, 2, \ldots, I$, and for development period $j$, $j = 1, 2, \ldots I$, let $C_{i,j}$ and $Y_{i,j}$ represent the total cumulative paid amount and the incremental paid amount, respectively with $Y_{i,j} = C_{i,j} - C_{i,j-1}$, $i = 1, \ldots, I$, $j = 2, \ldots, I$.

\begin{align*}
\begin{bmatrix}
    C_{1,1} & C_{1,2} & \ldots & C_{1,I-1} & C_{1,I}\\
    C_{2,1} & C_{2,2} & \ldots & C_{2,I-1} & \\
    \vdots & \vdots & \ddots & & \\
    C_{I,1}   &&&  
\end{bmatrix} 
\end{align*}

\begin{align*}
\begin{bmatrix}
    Y_{1,1} = C_{1,1} & Y_{1,2} & \ldots & Y_{1,I-1} & Y_{1,I}\\
    Y_{2,1} = C_{2,1} & Y_{2,2} & \ldots & Y_{2,I-1} & \\
    \vdots & \vdots & \ddots & & \\
    Y_{I,1} = C_{I,1}   &&&  
\end{bmatrix}
\end{align*}

where columns, rows and diagonals represent development, occurrence and calendar periods, respectively. Each incremental cell $Y_{i,j}$ can be seen as a cluster stacking $n_{i,j}$ amounts paid in the same development period $j$ for the occurrence period $i$. These payments come from $M$ claims and let $Y_{i,j}^{(k)}$ represent the sum of all observations from claims $k$ in the cluster $(i,j)$. It should be noted that all claims are not necessarily represented in each of the clusters.

To calculate a \emph{best estimate} for the reserve, the lower part of the triangle must be predicted and the total reserve amount is

\begin{align*}
\widehat{R} &= \sum_{t=2}^{I} \widehat{C}_{t, I} - \sum_{t=2}^{I} C_{t, I-t+1} = \sum_{t=2}^I\sum_{s=I+2-t}^I \widehat{Y}_{t,s}.
\end{align*}

To quantify uncertainty in estimated claims reserve, we consider the mean square error of prediction (MSEP). Let $\widehat{R}$ be a $\mathcal{Y}$-mesurable estimator for $\Esp{R|\mathcal{Y}}$ and a $\mathcal{Y}$-mesurable predictor for $R$ where $\mathcal{Y}$ represents the set of observed clusters. The MSEP is

\begin{align*}
    MSEP_{R|\mathcal{Y}}(\widehat{R}) &= \Esp{\left(\widehat{R} - R\right)^2|\mathcal{Y}}\\
        &= \Var{R|\mathcal{Y}} + \left(\widehat{R} - \Esp{R|\mathcal{Y}}\right)^2.
        \intertext{Independence between $R$ and $\mathcal{Y}$ is assumed, so the equation is simplified as follows}
    MSEP_{R|\mathcal{Y}}(\widehat{R}) &= \Var{R} + \left(\widehat{R} - \Esp{R}\right)^2
    \intertext{and the unconditional MSEP is}
    MSEP_{R}(\widehat{R}) &= \Var{R} + \Esp{\left(\widehat{R} - \Esp{R}\right)^2}.
\end{align*}

\subsection{The quasi-Poisson model for reserves}\label{ssec:qP}
\subsubsection{Construction}

From the theory presented in Subsection~\ref{ss:quasi}, we construct quasi-Poisson macro- and micro-level models for reserves. For both models, constitutive elements are defined in Table~\ref{tab:def1}.
\begin{table}[H]
\caption{Quasi-Poisson macro- and micro-level models for reserve ($i,j
  = 1, \ldots, I$). All clusters and all payments are independent.}
\centering
\begin{tabular}{lll}
\toprule    
    Components & Macro & Micro\\
\midrule    
    Exp. value & $\Esp{Y_{i,j}} = \lambda_{i,j}$ & $\Esp{Y_{i,j}^{(k)}} = \lambda_{i,j}$\\
    Inv. link func. & $\lambda_{i,j} = \exp[\bx_{i,j}\transpose\bb]$ & $\lambda_{i,j} = \exp[\bx_{i,j}\transpose\ba + \log(1/n_{i,j})]$\\
        & $\phantom{\lambda_{i,j}}= \exp[b_i + b_{I+j}]$& $\phantom{\lambda_{i,j}}= \exp[a_i + a_{I+j} + \log(1/n_{i,j})]$\\
        & with $b_{I+1} = 0$ & with $a_{I+1} = 0$\\
    Variance & $\Var{Y_{i,j}} = \varphi_{macro}\lambda_{i,j}$ & $\Var{Y_{i,j}^{(k)}} = \varphi_{micro}\lambda_{i,j}$\\
    Pred. value & $\widehat{Y}_{i,j} = \exp[\widehat{b}_i + \widehat{b}_{I+j}]$ & $\widehat{Y}_{i,j}^{(k)} = \exp[\widehat{a}_i + \widehat{a}_{I+j} + \log(1/n_{i,j})]$\\
    Known values & $\mathcal{Y}_{macro}$ & $\mathcal{Y}_{micro}$\\
\bottomrule
\end{tabular}
\label{tab:def1}
\end{table}
As a direct consequence of Proposition~\ref{propPoisson}, the best estimate for the total reserve amount is

\begin{align*}
\widehat{R} = \sum_{(i,j) \in \mathcal{K}}\sum_{k = 1}^{n_{i,j}} \widehat{Y}_{i,j}^{(k)} &= \sum_{(i,j) \in \mathcal{K}} \widehat{Y}_{i,j},
\end{align*}

where $\mathcal{K}$ represents unobserved clusters. For both models, the Proposition~\ref{prop:MSEP} gives results for the unconditional MSEP.
\begin{Proposition}\label{prop:MSEP}
In the quasi-Poisson macro-level model, the unconditional MSEP is given by

\begin{align*}
    \widehat{MSEP}_{R}(\widehat{R}) &\approx \sum_{(i,j) \in \mathcal{K}} \widehat{\varphi}_{macro}\widehat{y}_{i,j}\\
            &\phantom{=}+ \sum_{(i,j),(n,m) \in \mathcal{K}} \widehat{\varphi}_{macro}\widehat{y}_{i,j}\widehat{y}_{m,n}\bx_{i,j}\transpose\left(\bx \boldsymbol{W} \bx\transpose\right)^{-1} \bx_{n,m},
\end{align*}

where $\bx$ and $\boldsymbol{W}$ are defined by Equation~\eqref{eq:W}. The unconditional MSEP for the quasi-Poisson micro-level model is similar with $\widehat{\varphi}_{macro}$ replaced by $\widehat{\varphi}_{micro}$.
\end{Proposition}
\begin{proof}
The proof for the macro-level model is done in \cite{MW}. For the micro-level model, we have

\begin{align*}
 &MSEP_{R}(\widehat{R}) = \Var{R} + \Esp{\left(\widehat{R} - \Esp{R}\right)^2}\\
 &= \sum_{(i,j) \in \mathcal{K}} \sum_{k=1}^{n_{i,j}} \widehat{\varphi}_{micro} \exp[\bx_{i,j}\transpose\widehat{\ba} + \log(1/n_{i,j})]\\
    &\phantom{=}+ \sum_{(i,j) \in \mathcal{K}}\sum_{(m,n) \in \mathcal{K}}\sum_{k = 1}^{n_{i,j}}\sum_{t=1}^{n_{m,n}}\Cov{\widehat{Y}_{i,j}^{(k)}, \widehat{Y}_{m,n}^{(t)}}\\
        &= \sum_{(i,j) \in \mathcal{K}}  \widehat{\varphi}_{micro} \exp[\bx_{i,j}\transpose\widehat{\ba}]\\
    &\phantom{=}+ \sum_{(i,j) \in \mathcal{K}}\sum_{(m,n) \in \mathcal{K}}\sum_{k = 1}^{n_{i,j}}\sum_{t=1}^{n_{m,n}} \exp[\bx_{i,j}\transpose\ba + \log(1/n_{i,j})]\exp[\bx_{m,n}\transpose\ba + \log(1/n_{m,n})]\\
    &\phantom{=}\times\Cov{\exp[\bx_{i,j}\transpose\widehat{\ba} - \bx_{i,j}\transpose\ba], \exp[\bx_{m,n}\transpose\widehat{\ba} -\bx_{m,n}\transpose\ba]}\\
    \intertext{Although $\widehat{Y}_{i,j}^{(k)}$ is not an unbiased estimator of $\Esp{Y_{i,j}^{(k)}}$, the bias is generally of small order and by using the approximation $\exp[x] \approx 1 + x$ for $x \approx 0$, we obtain}
     &= \sum_{(i,j) \in \mathcal{K}}  \widehat{\varphi}_{micro} \exp[\bx_{i,j}\transpose\widehat{\ba}]\\
    &\phantom{=}+ \sum_{(i,j),(m,n) \in \mathcal{K}}\exp[\bx_{i,j}\transpose\ba+\bx_{m,n}\transpose\ba]\Cov{\bx_{i,j}\transpose\widehat{\ba}, \bx_{m,n}\transpose\widehat{\ba}}.
    \intertext{By using the fact that $\widehat{\bb} = \widehat{\ba}$ and the remark at the end of subsection~\ref{ss:quasi}, we obtain}
    &= \sum_{(i,j) \in \mathcal{K}}  \widehat{\varphi}_{micro} \widehat{y}_{i,j} + \sum_{(i,j),(m,n) \in \mathcal{K}} \widehat{\varphi}_{micro}\widehat{y}_{i,j}\widehat{y}_{m,n} \bx_{i,j}\transpose\left(\bx \boldsymbol{W}\bx\transpose\right)^{-1}\bx_{m,n}.
\end{align*}
\end{proof}
Thus, the difference between the variability in macro- and micro-level models results from the difference between dispersion parameters. Define standardized residuals for both models

\begin{align*}
    r_{i,g} &= \frac{\left(y_{i,g} - \widehat{y}_{i,g}\right)}{\sqrt{\widehat{y}_{i,g}}} \text{ and }
    r_{g} = \frac{\left(y_{g} - \widehat{y}_{g}\right)}{\sqrt{\widehat{y}_{g}}}.
\end{align*}

Direct calculations lead to

\begin{align}\label{eq:psi}
\Psi = \frac{\sum_{i,g} r_{i,g}^2}{\sum_g r_g^2}  &\leq \frac{\sum_g n_g - (k+1)}{m - (k+1)} \to \widehat{\varphi}_{\text{micro}} \leq \widehat{\varphi}_{\text{macro}}.
\end{align}

Thus, if the total number of payments ($\sum_g n_g$) is greater than the value $\Psi(m - (k+1)) + k + 1$, then the micro-level model \eqref{eq3} will lead to a greater precision for the best estimate of the total reserve amount and conversely. Adding one or more covariate(s) at the micro level will decrease the numerator of $\Psi$ and will increase the interest of the micro-level model.
\subsubsection{Illustration and Discussion}
To illustrate these results, we consider the incremental run-off triangle from UK Motor Non-Comprehensive account (published by \cite{Chris}) presented in Table~\ref{tab:triangle1} where each cell $(i,j)$, $i + j \leq 7$, is assumed to be a cluster $g$, i.e., the value $Y_g$ is the sum of $n_g$ independent payments. 
\begin{table}[H]
\caption{Incremental run-off triangle for macro-level model (in
  $000$'s).}
\centering
\begin{tabular}{llllllll}
\toprule    
   & $1$ & $2$ & $3$ & $4$ & $5$ & $6$ & $7$\\
\midrule    
$1$ & $\numprint{3511}$ &$\numprint{3215}$ &$\numprint{2266}$ &$\numprint{1712}$  &$\numprint{1059}$ & $\numprint{587}$ & $\numprint{340}$\\
$2$ & $\numprint{4001}$ &$\numprint{3702}$ &$\numprint{2278}$ &$\numprint{1180}$  &$\numprint{956}$& $\numprint{629}$ & --\\
$3$ & $\numprint{4355}$ &$\numprint{3932}$ &$\numprint{1946}$ &$\numprint{1522}$  &$\numprint{1238}$ & -- & --\\
$4$ & $\numprint{4295}$ &$\numprint{3455}$ &$\numprint{2023}$ &$\numprint{1320}$  &-- & -- & --\\ 
$5$ & $\numprint{4150}$ &$\numprint{3747}$ &$\numprint{2320}$ & -- &-- & -- & --\\
$6$ & $\numprint{5102}$ &$\numprint{4548}$ & -- & -- &-- & -- & --\\
$7$ & $\numprint{6283}$ &-- & -- & --&-- & -- & --\\
\bottomrule
\end{tabular}
\label{tab:triangle1}
\end{table}
We construct $2$ macro-level models

\begin{align*}
\text{\textbf{Model A}: }Y_{g} &\sim \mathcal{P}(\lambda_{g}) \qquad \lambda_{g} = \exp[\bx_{g}\transpose\ba]\\
\text{\textbf{Model B}: }Y_{g} &\sim q\mathcal{P}(\lambda_{g})\intertext{and $2$ micro-level models}
\text{\textbf{Model C}: }Y_{i,g} &\sim \mathcal{P}(\lambda_{i,g}) \qquad \lambda_{i,g} =\exp[\bx_{g}\transpose\ba - \log(n_g)]\\
\text{\textbf{Model D}: }Y_{i,g} &\sim q\mathcal{P}(\lambda_{i,g}).
\end{align*}

The final reserve amount obtained from the Mack's model (\cite{Mack}) is $\numprint{28655773}\$$. To create micro-level datasets from the "macro" one, we perform the following procedure:
\begin{enumerate}%[leftmargin=*,labelsep=3mm]
\item simulate the number of payments for each cluster assuming $N_g \sim \mathcal{P}(\theta)$, $g = 1, \ldots, m$;
\item for each cluster, simulate a $(n_g \times 1)$ vector of proportions assuming $\boldsymbol{\omega}_g = \begin{bmatrix} \omega_1 & \ldots & \omega_{n_g}\end{bmatrix}\transpose \sim \text{Dirichlet}(\boldsymbol{1})$, $g = 1, \ldots, m$; 
\item for each cluster, define 

\begin{align*}
    \begin{bmatrix} Y_{1,g} \\ \vdots \\ Y_{n_g, g} \end{bmatrix} &= \floor{\boldsymbol{\omega}_gY_g}, \qquad g = 1,\ldots, m;
\end{align*}

\item adjust \textbf{Model C} and \textbf{Model D}; and
\item calculate the best estimate and the MSEP of the reserve.
\end{enumerate}
For each value of $\theta$, we repeat this procedure $\numprint{1000}$ times and we calculate the average best estimate and the average MSEP. 
Those results are consistent with Corollaries \ref{cor:1}, \ref{cor:2} and Proposition \ref{prop:MSEP}. For Poisson regression (\textbf{Model A} and \textbf{C}), results are similar. For micro-level models, convergence of $\sqrt{MSEP}$ fowards ($\numprint{11622}$) is fast. For quasi-Poisson regression (\textbf{Model B} and \textbf{D}), Figure~\ref{fig:res1} shows $\sqrt{MSEP}$ as a function of a expected total number of payments, for the portfolio. Above a certain level, (close to $\numprint{3400}$ here), accuracy of the "micro" approach exceed the "macro". 
 \begin{figure}[H]
    \centering
    \includegraphics[height=7cm,width=7cm]{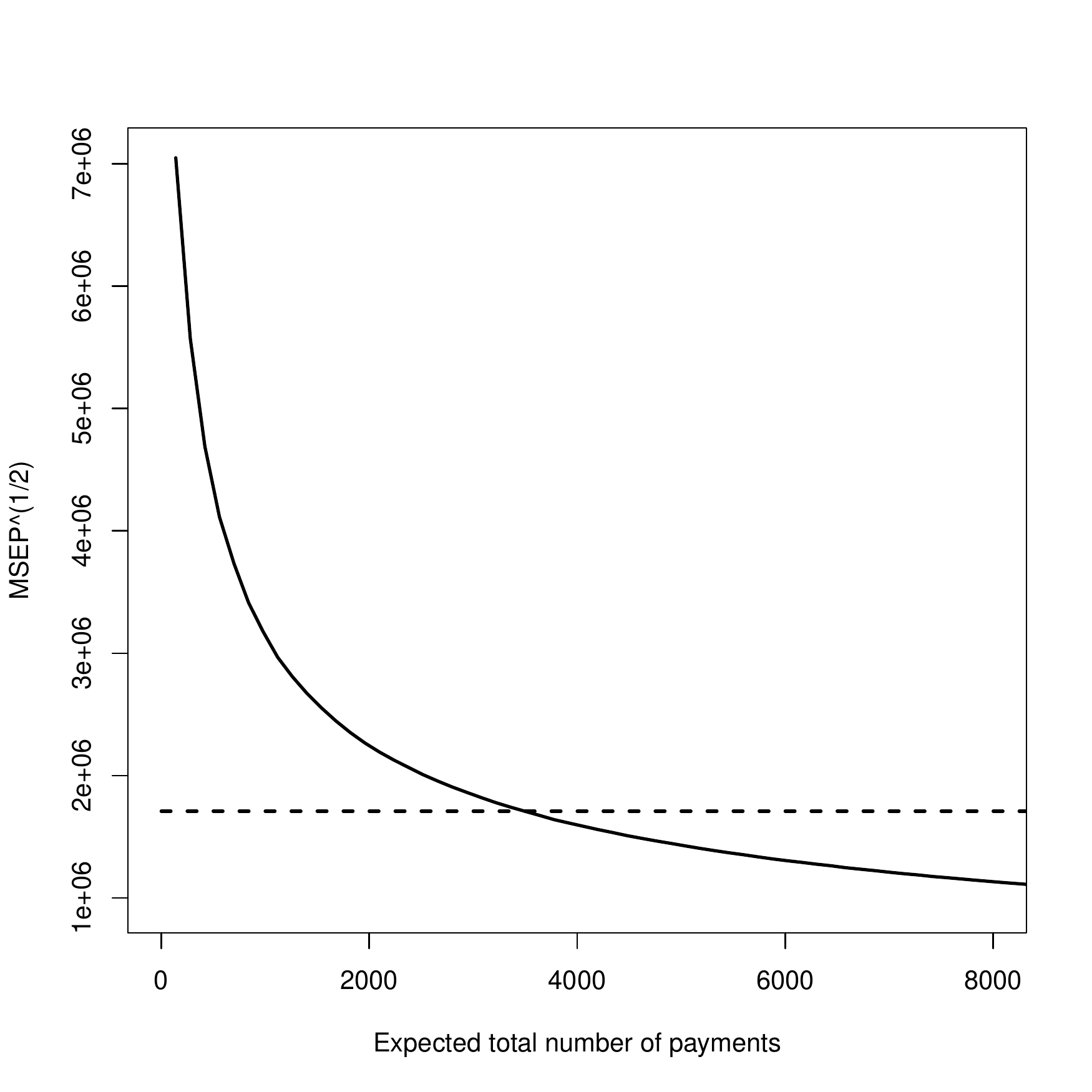}
    \caption{Square root of the mean square error of prediction obtained for \textbf{Model D} (solid line) and the \textbf{Model C} (broken line) from simulated values for increasing expected number of payments for the portfolio.}
    \label{fig:res1}
  \end{figure}
  
In order to illustrate the impact of adding a covariate at the
micro-level, we define a quasi-Poisson micro-level model with a weakly
correlated covariate (\textbf{Model E}) and with a strongly correlated
covariate (\textbf{Model F}). Following a similar procedure, we obtain
results presented in Table~\ref{tab:res1} and Figure~\ref{fig:res2}. 
\begin{table}[H]
\caption{Results.}
\centering
\begin{tabular}{lll}
\toprule    
Method & $\Esp{\text{Reserve}}$ & $\sqrt{MSEP}$\\
\midrule    
Mack's model & $\numprint{28655773}$ & $\numprint{1417267}$\\
\midrule
Poisson reg. & &\\
\textbf{Model A} & $\numprint{28655773}$ & $\numprint{11622}$\\
\textbf{Model C} & $\numprint{28655773}$ & $\numprint{11622}$\\
\midrule
quasi-Poisson reg. & &\\
\textbf{Model B} & $\numprint{28655773}$ & $\numprint{1708196}$\\
\textbf{Model D} & $\numprint{28655773}$ & see Figure \ref{fig:res1}\\
\midrule
\midrule
quasi-Poisson reg. & &\\
\textbf{Model E} ($\rho \approx 0$) & $\numprint{28657364}$ & see Figure \ref{fig:res2}\\
\textbf{Model F} ($\rho \approx 0.8$) & $\numprint{20514566}$ & see Figure \ref{fig:res2}\\
\bottomrule
\end{tabular}
\label{tab:res1}
\end{table}
As opposed to standard classical results on hierarchical models, the average of explanatory variable within a cluster ($(1/n_g)\sum_i x_{ig}$) has not been added to the macro-level model (\textbf{Model B}), for several reasons,
\begin{itemize}%[leftmargin=*,labelsep=4mm]
    \item[\textbf{(i)}] impossible to compute that average without individual data;
    \item[\textbf{(ii)}] discrete explanatory variables used in the micro-level model; and
    \item[\textbf{(iii)}] since claims reserve model have a predictive motivation, it is risky to project the value of an aggregated variable on future clusters.
\end{itemize}
\begin{figure}[H]
    \centering
    \includegraphics[height=7cm,width=7cm]{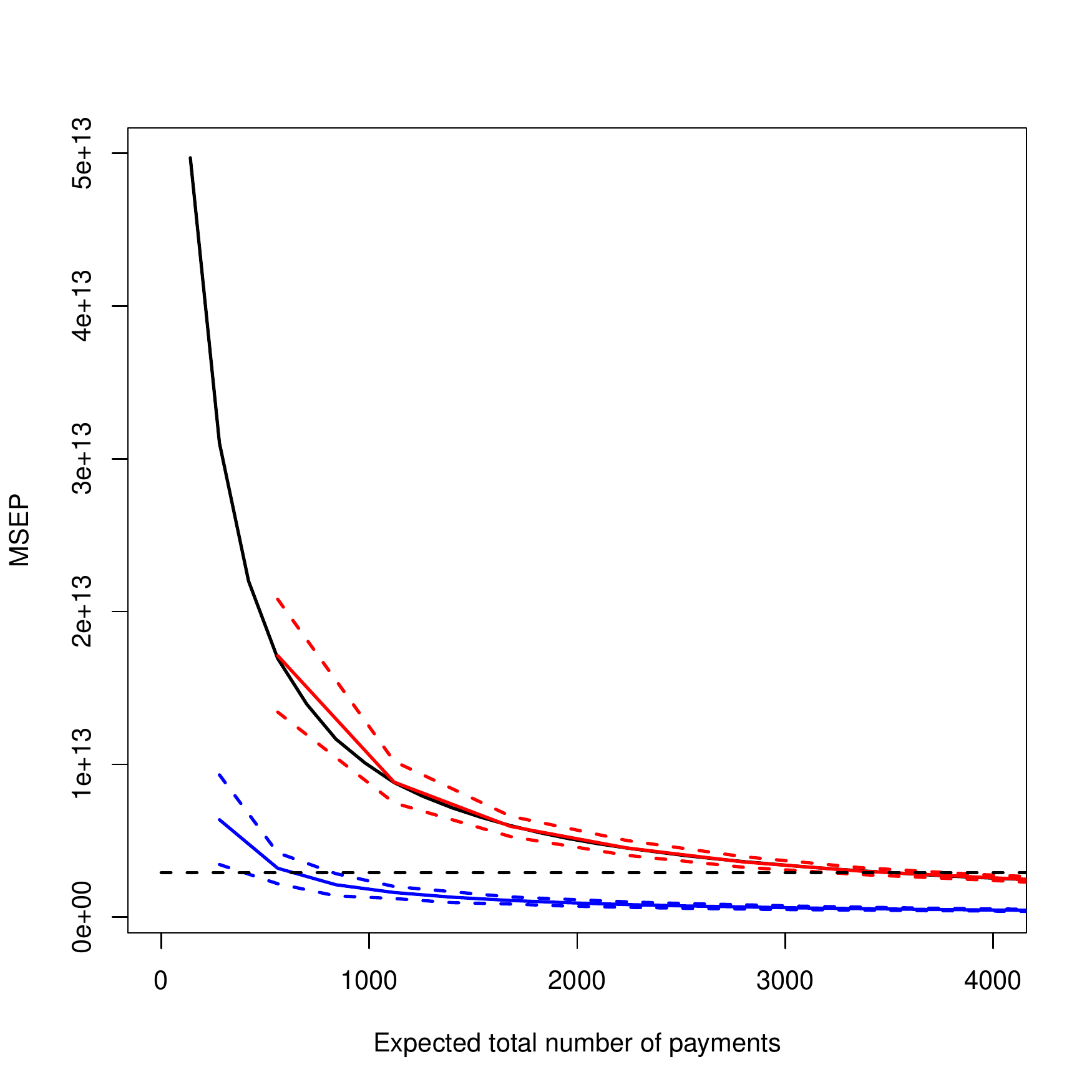}
    \caption{Mean square error of prediction ($\pm 2\sigma$) obtained from simulated values as a function of the expected number of payments for \textbf{Model E} (red lines) and \textbf{Model F} (blue lines). For comparison purposes, the MSEP obtained for the \textbf{Model D} (solid black line) and the \textbf{Model B} (broken black line) are added.}
    \label{fig:res2}
  \end{figure}
With an explanatory variable highly correlated with the response variable, results obtained with \textbf{Model D} and \textbf{E} are very close. As claimed by Proposition \ref{prop:MSEP} and equation \eqref{eq:psi}, an explanatory variable highly correlated with the response variable will decrease the value of $\sqrt{MSEP}$, and lowers the threshold above which the micro-level model is more accurate than the macro-level one.

The quasi-Poisson macro-level model (\textbf{Model B}) with maximum likelihood estimators leads to the same reserves as the chain-ladder algorithm and the Mack's model (see \cite{MackVenter}), assuming the clusters exposure, for $(i,j) \in \mathcal{K}$, is one. To obtain similar results with a quasi-Poisson micro-level model (\textbf{Model D}), a similar assumption is necessary: exposure of each claim within cluster $(i,j)$ is $1/n_{i,j}$. That assmption implies, on a micro level, that predicted individual payments $\widehat{Y}_{ij}^{(k)}$ are proportional to $1/n_{ij}$. That assumption has unfortunately no foundation. 

In the Poisson and quasi-Poisson micro-level models (\textbf{Model C} and \textbf{D}),
payments related to the same claim, in two different clusters are
supposed to be non-correlated. As discussed in the previous Section,
it is possible to include dependencies among payments for a given
claim using a Poisson regression with random effects. Simulations and computations were performed in \textsf{R}, using packages \textsf{ChainLadder} and \textsf{gtools}. 

\subsection{The Mixed Poisson model for reserves}

\subsubsection{Construction}

From the results obtained in Section \ref{ss:RE}, it is possible to
construct a micro-model for the reserves that includes a random intercept.The later will allow to model dependence between payments from a given claim.
Note that it is hard to find an aggregated model with random effects that could be compared with individual ones.
In the context of claims reserves, $Y_g^{(t)}$
represents the sum of paids made for claim $t$ within cluster $g$. The
assumptions of that model (called \textbf{model G}) are

  \begin{align*}
    \left(Y_g^{(t)} | \gamma_t\right) &\sim \mathcal{P}(\lambda_{g}e^{\gamma_t}), \qquad
    \lambda_{g} = \exp[\boldsymbol{x}_g\transpose\boldsymbol{c}
    +ln(1/n_g)]\\
    \gamma_t &\sim N(0, \sigma^2).
  \end{align*}

Because of those two random variables in the model, two kinds of
predictions can be derived: un-conditional ones, where

  \begin{align*}
    \left(\widehat{Y}_g^{(t)}| \gamma_t\right) &\sim
    \mathcal{P}(\widehat{\lambda}_ge^{\gamma_t})\\
    \widehat{\lambda}_g &= \exp[\boldsymbol{x}_g\transpose\widehat{\boldsymbol{c}}+ln(1/n_g)]\\
    \text{ so that } \Esp{\widehat{Y}_g^{(t)}} &= \lambda_ge^{\sigma^2/2};
  \end{align*}

and conditional ones, where the unknown magnitude of claim $t$ is
predicted by the so-called best linear estimate (that minimizes the
MSEP) $\tilde{\gamma}_t$ (see \cite{MW})

  \begin{align*}
   \left(\tilde{Y}_g^{(t)}| \tilde{\gamma}_t\right) &\sim
    \mathcal{P}(\widehat{\lambda}_ge^{\tilde{\gamma}_t})\\
    \text{ so that } \Esp{\tilde{Y}_g^{(t)}} &= \lambda_ge^{\tilde{\gamma}_t}.
  \end{align*}

It is then possible to compute the overall best estimate for the total amount of reserves.

\subsubsection{Illustration and Discussion}

In order to construct a micro-level model from triangle \ref{tab:triangle1}, we follow a procedure to the one described in the provious section, with steps 1--3 (that are not mentioned here)
  \begin{enumerate}%[leftmargin=*,labelsep=3mm]
  \item[4.] for each accident year, allocate randomly the source of each payment
  \item[5.] fit \textbf{model G}; and 
  \item[6.] compute the best estimate and the MSEP of the reserve.
  \end{enumerate}
For a fixed value of $\theta$, the procedure is repeated $1000$
times. Various values were considered for $\theta$ ($10$, $25$, $50$, $100$ and $250$),  and results were similar. In order to avoid heavy tables, only the case where $\theta = 10$ is mentioned here. Simulations and computations were performed with \textsf{R}, relying on package \textsf{lme4}. On Figure \ref{fig:val1} we can see predictions of the model on observed data, while on Figure \ref{fig:val2} we can see predictions of the model for non-observed cells. Finaly, results are reported in Table \ref{tab:val1}. At each step, a LRT is performed (see section \ref{ss:RE}) and each time, the variance at origin was significant non-null, meaning that correlation among payments (related to the same claim) is positive. Observe that with the random model, the log-likelihood is approximated using numerical intergration, which might bias computed $p$-values of the test. Here, $p$ have been confirmed using a bootstrap procedure (using package \textsf{glmmML}).
\begin{figure}[H]
    \centering
    \includegraphics[height=7cm,width=7cm]{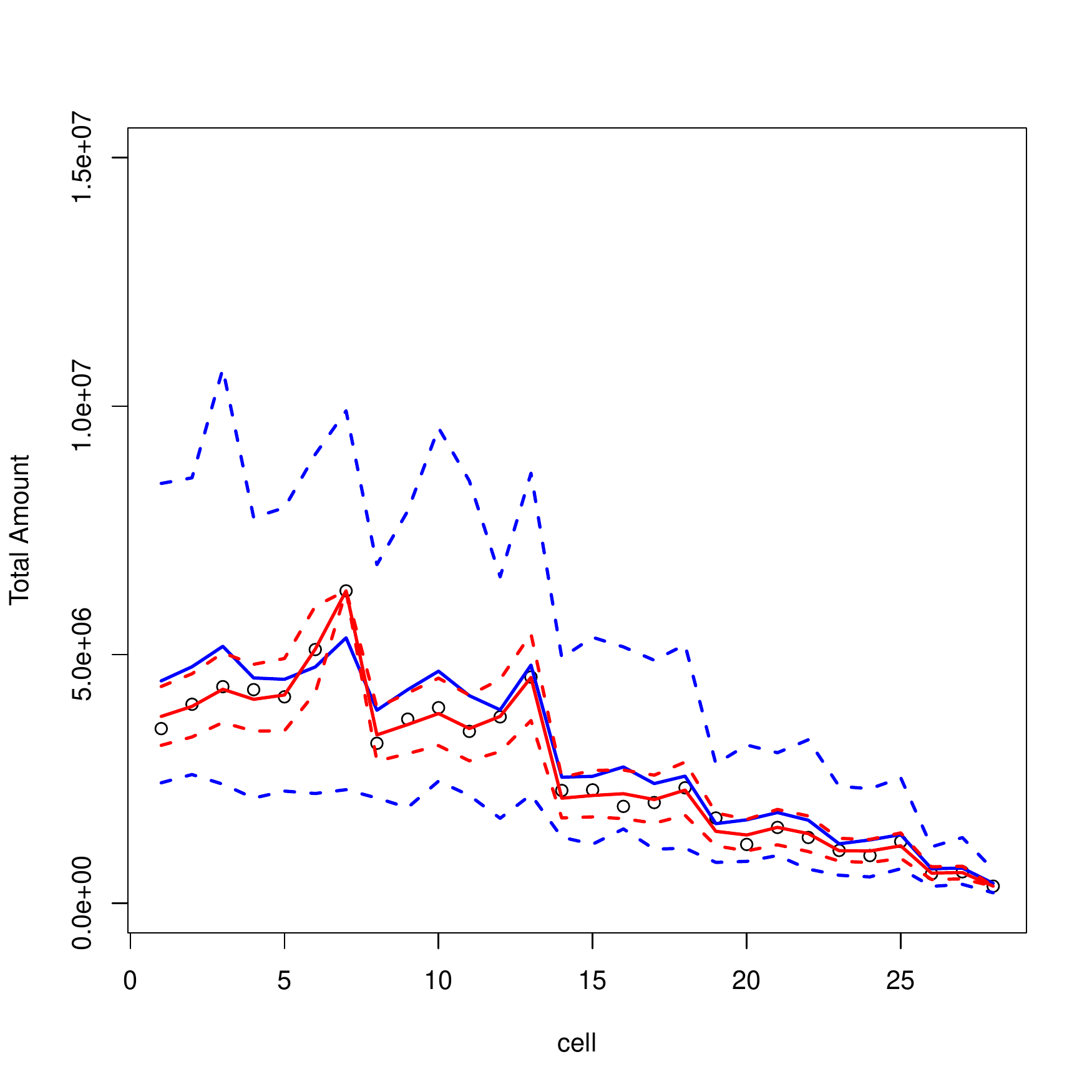}
    \caption{Observed data (circles) with conditional prediction
      (plain red lines) and un-conditional (plain blue lines) from
      \textbf{model G} ($\theta = 10$).}
    \label{fig:val1}
  \end{figure}
\begin{figure}[H]
    \centering
    \includegraphics[height=7cm,width=7cm]{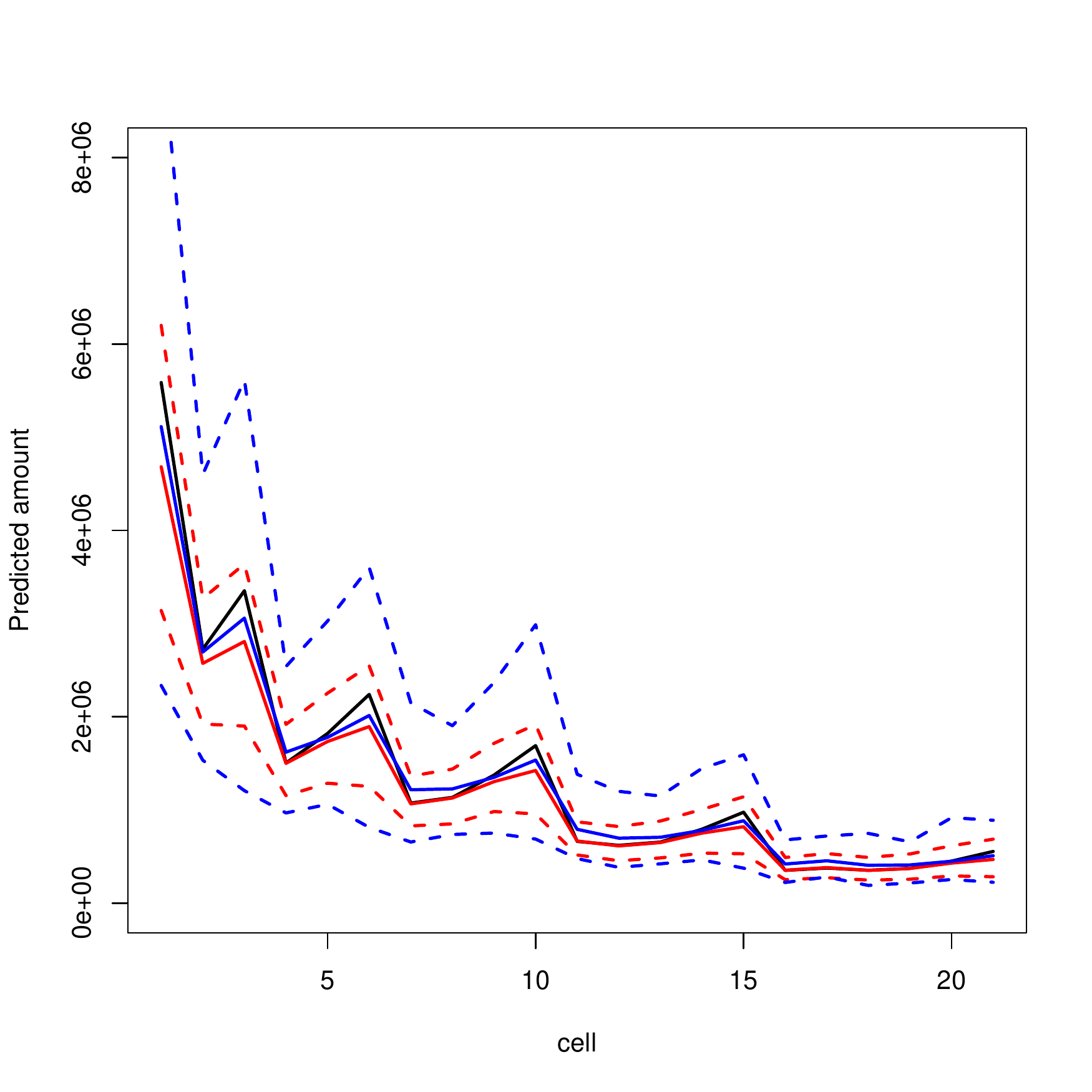}
    \caption{Predictions with the quasi-Poisson macro-level model (strong
      black line), with conditional prediction (plain red lines) and
      un-conditional (plain blue lines) from \textbf{model G} ($\theta = 10$).}
    \label{fig:val2}
  \end{figure}
\begin{table}
\caption{Numerical Results for $\theta = 10$. Results for different
  values of $\theta$ are similar.}
\centering
\begin{tabular}{lll}
\toprule    
Modèle & $\Esp{\text{Reserve}}$ & $\sqrt{\text{Var(Reserve)}}$\\
\midrule 
coll. quasi-Pois. & $\numprint{28656423}$ & $\numprint{1708216}$\\
mixed Poisson non-cond. & $\numprint{27930624}$ & $\numprint{3297401}$\\
mixed Poisson cond. & $\numprint{25972947}$ & $\numprint{2280902}$\\
\bottomrule
\end{tabular}
\label{tab:val1}
\end{table}

%%%%%%%%%%%%%%%%%%%%%%%%%%%%%%%%%%%%%%%%%%
%\acknowledgments{The first author received financial support from NSERC and the ACTINFO chair. The second author received financial support from NSERC.}

%%%%%%%%%%%%%%%%%%%%%%%%%%%%%%%%%%%%%%%%%%

%\authorcontributions{For research articles with several authors, a short paragraph specifying their individual contributions must be provided. The following statements should b%e used ``X.X. and Y.Y. conceived and designed the experiments; X.X. performed the experiments; X.X. and Y.Y. analyzed the data; W.W. contributed reagents/materials/analysis too%ls; Y.Y. wrote the paper.'' Authorship must be limited to those who have contributed substantially to the work reported.}

%%%%%%%%%%%%%%%%%%%%%%%%%%%%%%%%%%%%%%%%%%

%\conflictofinterests{The authors declare no conflict of interest.} 

%%%%%%%%%%%%%%%%%%%%%%%%%%%%%%%%%%%%%%%%%%
%\bibliographystyle{mdpi}

%=====================================
% References, variant A: internal bibliography
%=====================================
%\renewcommand\bibname{References}

%%%%%%%%%%%%%%%%%%%%%%%%%%%%%%%%%%%%%%%%%%

\end{document}